\documentclass[11pt]{article}

\usepackage[margin=1in]{geometry}
\usepackage[numbers,sort&compress]{natbib}
\usepackage{amsmath}
\usepackage{amssymb}
\usepackage{amsthm}
\usepackage{algorithm}
\usepackage{algorithmic}
\usepackage{booktabs}
\usepackage{enumitem}
\usepackage{multicol}
\usepackage{xcolor}
\usepackage{soul}
\usepackage{graphicx}
\usepackage[small]{caption}
\usepackage{microtype}
\usepackage{hyperref}
\usepackage{doi}
\usepackage{tikz}
\usepackage{eufrak,bbold}
\hypersetup{
    colorlinks=true,
    linkcolor=black,
    citecolor=black,
    filecolor=black,
    urlcolor=[HTML]{0088CC},
    pdftitle={},
}

\theoremstyle{definition}

\newtheorem{prop}{Proposition}[section]
\newtheorem{lthm}[prop]{Theorem}
\newtheorem{lem}[prop]{Lemma}
\newtheorem{cor}[prop]{Corollary}

\newtheorem{rem}[prop]{Remark}

\newtheorem{eg}[prop]{Example}
\newtheorem{dfn}[prop]{Definition}

\newtheorem{obs}[prop]{Observation}

\newcommand{\dft}[1]{\textbf{\textit{#1}}}

\newcommand{\abs}[1]{\left|#1\right|}
\newcommand{\paren}[1]{\left(#1\right)}
\newcommand{\set}[1]{\left\{#1\right\}}
\newcommand{\sucht}{\,\middle|\,}

\newcommand{\N}{\mathbf{N}}

\newcommand{\bfP}{\mathbf{P}}

\newcommand{\NP}{\mathrm{NP}}
\newcommand{\Poly}{\mathrm{P}}

\newcommand{\shP}{\mathrm{\#P}}

\newcommand{\CountCliques}{\mathsf{CountCliques}}
\newcommand{\CountExtensions}{\mathsf{CountExtensions}}
\newcommand{\CountValidLabelings}{\mathsf{CountValidLabelings}}

\newcommand{\SampleLabeling}{\mathsf{SampleLabeling}}

\DeclareMathOperator{\pw}{pw}

\DeclareMathOperator{\width}{width}

\DeclareMathOperator{\range}{range}

\newcommand{\calD}{\mathcal{D}}

\newcommand{\calI}{\mathcal{I}}
\newcommand{\calL}{\mathcal{L}}

\newcommand{\calX}{\mathcal{X}}
\newcommand{\calY}{\mathcal{Y}}

\title{Simple Counting and Sampling Algorithms for Graphs with Bounded Pathwidth}

\author{
  Christine T.\ Cheng\\
  University of Wisconsin, Milwaukee\\
  ccheng@uwm.edu
  \and
  Will Rosenbaum\\
  Amherst College\\
  wrosenbaum@amherst.edu
}      

\date{\today}

\begin{document}

\maketitle
\thispagestyle{empty}

\begin{abstract}
  In this paper, we consider the problem of counting and sampling structures in graphs. We define a class of ``edge universal labeling problems''---which include proper $k$-colorings, independent sets, and downsets---and describe simple algorithms for counting and uniformly sampling valid labelings of graphs, assuming a path decomposition is given. Thus, we show that several well-studied counting and sampling problems are fixed parameter tractable (FPT) when parameterized by the pathwidth of the input graph. We discuss connections to counting and sampling problems for distributive lattices and, in particular, we give a new FPT algorithm for exactly counting and uniformly sampling stable matchings.
\end{abstract}

\section{Introduction}
\label{sec:introduction}

Pioneering work in computational complexity sought to separate computational problems broadly into ``easy'' and ``hard'' problems, for example separating $\NP$-hard problems from those in $\Poly$. As many problems that we would like to solve in practice are ($\NP$-)hard, researchers were led towards understanding which \emph{instances} of a problem are hard. Are there families of instances on which a hard problem becomes tractable? How does the \emph{structure} of an input affect the problem's complexity? This line of inquiry lead to the study of \emph{parameterized complexity}, which seeks to understand how various parameters of an input determine the efficiency with which a computational task can be performed.

Two landmark results in parameterized complexity are due to Courcelle~\cite{Courcelle1990-monadic} and Bodlaender~\cite{Bodlaender1996-linear-time}, respectively. In~\cite{Courcelle1990-monadic}, Courcelle proves a large class of (decision) problems on graphs---those definable in monadic second order logic---can be solved in linear time, assuming a ``tree decomposition of bounded width'' of the input is provided. Later, Bodlaender~\cite{Bodlaender1996-linear-time} provided an algorithm that \emph{finds} a tree decomposition of width $k$ of any graph $G$ in time $O(f(k) n)$, or reports that no such decomposition exists. Together, the results of Courcelle and Bodlaender show that a huge swath of graph problems are fixed parameter tractable (FPT) when parameterized by the tree-width of the graph (i.e., the minimal $k$ for which there exists a tree decomposition of width $k$).

Courcelle's result focused on decision problems, but Arnborg, Lagergren, and Seese in~\cite{Arnborg1991-easy} extended the results to counting problems as well. We refer the reader to~\cite{Curticapean2019-counting} for a recent survey of work on the parameterized complexity of counting. The results of Courcelle and Arnborg et al.\ are general, deep, and beautiful. Yet the generality of these results comes at the cost of accessibility. 

In this paper, we describe FPT algorithms for counting and sampling four familiar and well-studied graph structures: cliques, proper $c$-colorings, independent sets, and downsets. (The latter three structures are examples of ``edge universal labeling problems''---cf.\ Definition~\ref{dfn:edge-universal}---a class of problems for which we describe a generic-yet-simple algorithm.) Counting cliques, $c$-colorings, independent sets, and downsets are all known to be $\shP$-complete (see, e.g.,~\cite{Dyer2004-relative} and references therein), hence $\NP$-hard. Given a path decomposition of a graph, our algorithms employ a straightforward application of dynamic programming to exactly count and uniformly sample structures in a graph. Thus, our results give non-trivial examples of hard counting problems that can be solved exactly in near-linear time for restricted graph families. We believe the algorithms and concepts are simple enough to be accessible to students in an undergraduate algorithms course.

We discuss applications of our algorithms to counting and sampling in distributive lattices, and in particular, for the stable marriage problem. It is well-known that the set of stable matchings for an instance forms a distributive lattice. The stable matchings can be succinctly represented as the family of downsets in the ``rotation digraph'' of an instance, which itself can be computed in near-linear time. In a companion paper~\cite{Cheng2020-stable}, we show that for a natural parameterization---the ``$k$-range model'' introduced by Bhatnagar et al~\cite{Bhatnagar2008-sampling}---instances have rotation digraphs whose pathwidths are bounded by a function of range of the instance, and that a path decomposition of the rotation digraph can be computed efficiently from the input. Combined with the structural result in~\cite{Cheng2020-stable}, our algorithms for counting and sampling downsets give FPT algorithms for counting and sampling stable matchings parameterized by the ``range'' of the input. Specifically, for any fixed constant $k$ and $k$-range stable marriage instance, it is possible to exactly count and uniformly sample stable matchings in linear time (Corollary~\ref{cor:sm}). This result is in contrast to the work of Bhatnagar et al.~\cite{Bhatnagar2008-sampling}, who showed that a natural Markov chain Monte Carlo approach to sampling stable matchings requires exponential time in the $k$-range model for any $k \geq 5$.


\section{Pathwidth and Cliques}
\label{sec:pathwidth}

Throughout the paper, we let $G = (V, E)$ denote a directed or undirected graph. We use $uv$ to denote a directed or undirected edge between $u$ and $v$. We begin by briefly reviewing some fundamental results regarding pathwidth. The section concludes with extremely simple algorithms for counting and sampling cliques in a graph.

\begin{dfn}
  \label{dfn:path-decomp}
  A \dft{path decomposition} of graph $G = (V, E)$ is a sequence $(X_1, X_2, \ldots, X_r)$ of subsets of $V$ such that:
  \begin{enumerate}
  \item $\bigcup_{i = 1}^r X_i = V$,
  \item for each edge $uv \in E$, there exists $i \in [r]$ such that $u, v \in X_i$,
  \item for all $i, j, k \in [r]$ with $i \leq j \leq k$, we have $X_i \cap X_k \subseteq X_j$.
  \end{enumerate}
  The \dft{width} of the path decomposition is $\width(\calX) = \max_{i} \abs{X_i} - 1$. The \dft{pathwidth} of $G$, denoted $\pw(G)$, is the minimum width over all path decompositions of $G$.
\end{dfn}

We extend the definition of pathwidths to directed graphs. 

\begin{dfn}
  \label{dfn:pathwidth-p}
  Let $H$ be a directed graph. The \dft{pathdwidth of $H$} is the pathwidth of the undirected version of $H$---that is, the undirected graph formed by replacing each directed edge in $H$ with an undirected edge with the same endpoints. 
\end{dfn}

\begin{rem}
  \label{rem:pathwidth-interval}
  Suppose $(X_1, X_2, \ldots, X_r)$ is a path decomposition of $G$. Item~3 above implies that for each vertex $v$, there is an interval $I_v \subseteq [r]$ such that $v \in X_i$ if and only if $i \in I_v$. By item~2, if $uv \in E$, then we must have $I_u \cap I_v \neq \varnothing$. Thus, $G$ is a subgraph of the interval graph\footnote{Recall that an \dft{interval graph} on a family $\calI$ of intervals is the graph $G = (\calI, E)$ where $\set{I, J} \in E$ if and only if $I \cap J \neq \varnothing$.} defined by the intervals $\set{I_v \sucht v \in V}$. For $I_v = [i_v, j_v]$, we say that $v$ is \dft{added} to the decomposition at index $i_v$, and \dft{removed} at index $j_v + 1$. 
\end{rem}

\begin{dfn}
  \label{dfn:nice-path} 
  Let $\calX = (X_1, X_2, \ldots, X_r)$ be a path decomposition of graph $G$. We say that $\calX$ is a \dft{nice path decomposition} if $\abs{X_1} =  1$, $\abs{X_r} = 0$ and for all $i \in [r - 1]$, we have $\abs{X_i\,\triangle\,X_{i+1}} = 1$. That is, when $\calX$ is nice, exactly one vertex is added or removed at each index.   
\end{dfn}

The (omitted) proof of the following lemma is straightforward.

\begin{lem}
  \label{lem:nice-path} Let $G$ be a graph with $n$ vertices. 
  Suppose $\calX = (X_1, X_2 \ldots, X_r)$ is a path decomposition of $G$ of width $k$. Then $G$ has a nice path decomposition $\calY = (Y_1, Y_2, \ldots, Y_{s})$ of width $k$ with $s = 2n$.    Moreover, $\calY$ can be computed from $\calX$ in time $O(k n)$. 
\end{lem}

The following seminal result of Bodlaender shows that computing the pathwidth and optimal path decompositions of a graph is fixed parameter tractable.

\begin{lthm}[Bodlaender~\cite{Bodlaender1996-linear-time}]
  \label{thm:pathwidth}
  Let $G$ be a graph and let $k \in \N$ be a constant.  There is an algorithm that decides whether $\pw(G) \leq k$ in $O(f(k)\abs{G})$ time. If $\pw(G) \leq k$, then the algorithm outputs a path decomposition $\calX$ of $G$ of width $k$.
\end{lthm}

The next corollary is immediate from Theorem~\ref{thm:pathwidth} and Lemma~\ref{lem:nice-path}.

\begin{cor}
  \label{cor:nice-path}
  For any graph $G$, a nice path decomposition of $G$ can be computed in time $O(f(k) \abs{G})$ where $k = \pw(G)$ and $f$ is some function depending only on $k$.
\end{cor}

\subsection{Counting and Sampling Cliques}


As a warm-up, we describe simple algorithms for counting and sampling cliques in a graph $G$, given a path decomposition $\calX = (X_1, X_2, \ldots, X_{2n})$ of $G$ of width $\pw$.

\begin{obs}
  \label{obs:clique}
  If $K \subseteq V$ is a clique in $G$, then there exists an index $i$ such that $K \subseteq X_i$. To see this, for each $v \in V$, let $I_v \subseteq [2n]$ denote the interval of indices $i$ such that $v \in X_i$. By Remark~\ref{rem:pathwidth-interval}, for every pair $v, w \in K$, $I_v$ and $I_w$ intersect. Therefore, the family of intervals $\{I_v, v \in K\}$ are mutually intersecting. That is, there exists $i \in \bigcap_{v \in K} I_v$ so that $K \subseteq X_i$.
\end{obs}

By Observation~\ref{obs:clique}, it is enough to count cliques in each (subgraph induced by) $X_i$. The only subtlety is that we must ensure that we only count each clique once (though it may be a subgraph of many $X_i$). To this end, we associate each clique $K$ with a single vertex $v \in K$---the last vertex in $K$ added in the path decomposition $\calX$. We describe the algorithm in pseudo-code below.

\begin{algorithm}
  \label{alg:count-cliques}
  \caption{$\CountCliques(G, \calX)$. $G = (V, E)$ is a graph, and $\calX = (X_1, X_2, \ldots, X_{2n})$ is a path decomposition of $G$ of width $\pw$. Let $I = \set{i_1, i_2, \ldots, i_n}$ be the set of indices at which vertices $v_1, v_2, \ldots, v_n$ (respectively) are inserted.}
  \begin{algorithmic}[1]
    \FORALL{$i \in I$}
    \STATE $C(v_i) \leftarrow$ number of cliques in $X_i$ containing $v_i$ \COMMENT{compute by brute force} \label{ln:clique-count}
    \ENDFOR
    \RETURN $\sum_{i = 1}^n C(v_i)$
  \end{algorithmic}
\end{algorithm}

\begin{lthm}
  \label{thm:count-cliques}
  Given a graph $G$ and a nice path decomposition $\calX$ of $G$ of width $\pw$, $\CountCliques(G, \calX)$ returns the number of cliques in $G$ in time $O(2^{\pw} \pw^2 n)$. Therefore, counting cliques can be done in FPT linear time, parameterized by the pathwidth of $G$.
\end{lthm}
\begin{proof}
  Let $K$ be a clique. Then $K$ is included in the count in Line~\ref{ln:clique-count} at iteration $i$ if and only if $v_i \in K$ and $K \subseteq \set{v_1, v_2, \ldots, v_i}$. Therefore, each clique $K$ is counted exactly once, so that $\CountCliques$ returns the correct value.

  For the runtime of $\CountCliques$, note that for all $i$, we have $\abs{X_i} \leq \pw + 1$. Finding the number of cliques containing $v_i$ in $X_i$ can be done by brute force in time $O(2^{\pw} \pw^2)$ by enumerating all $2^{\pw}$ subsets containing $v_i$ and checking if each is a clique in time $O(\pw^2)$. The overall runtime follows because these counts are made for each vertex exactly once.

  The final assertion of the theorem follows by applying Bodlaender's algorithm (cf.\ Corollary~\ref{cor:nice-path}) to obtain a path decomposition $\calX$ of $G$ whose width is the pathwidth of $G$.
\end{proof}

Once we have the count $C(v)$ for each vertex $v \in V$ from $\CountCliques$, sampling a uniformly random clique is straightforward: pick a random vertex, choosing $v_i$ with probability $C(v_i) / \sum_{j \in I} C(v_j)$. Then enumerate all cliques in $X_i$ containing $v_i$ (in time $O(2^{\pw} \pw^2)$), and return one uniformly at random. As a result, we obtain the following.

\begin{cor}
  \label{cor:sample-cliques}
  Given a graph $G$ and nice path decomposition $\calX$ of $G$ of width $\pw$, we can sample a clique from $G$ uniformly at random in time $O(2^{\pw} \pw^2 n)$. Therefore, a clique can be uniformly sampled from $G$ in FPT linear time, parameterized by the pathwidth of $G$.
\end{cor}

\section{Edge Universal Labeling Problems}

In this section, we describe a family of ``edge universal labeling problems,'' the family of graph problems we consider for the remainder of the paper. Algorithms for counting and sampling solutions to any edge universal labeling problem are described in Subsections~\ref{sec:counting-algorithms} and~\ref{sec:sampling-algorithms}, respectively.

A \dft{(vertex) labeling} of $G$ is a function $L \colon V \to \Sigma$, where $\Sigma$ is a finite set of \dft{labels}. Given a subset $W \subseteq V$, we denote the restriction of $L$ to $W$ by $L\restriction_W \colon W \to \Sigma$.

A \dft{partial labeling} is a function $K \colon V \to \Sigma \cup \set{\perp}$. In a partial labeling, a vertex $v$ satisfying $K(v) = \perp$ is said to be \dft{unassigned}, and \dft{assigned} otherwise. Given a partial labeling $K$, unassigned vertex $v$, and $\sigma \in \Sigma$, we use the notation $K \cup \set{v \mapsto \sigma}$ to denote the (partial) labeling that is equal to $K$ except that it maps $v$ to $\sigma$. We call a labeling $L \colon V \to \Sigma$ an \dft{extension} of $K$ if $L(v) = K(v)$ for all $v$ such that $K(v) \neq \perp$. 

A \dft{vertex labeling problem} $\calL$ is a family of pairs, $\calL = \set{(G, L)}$. For a fixed vertex labeling problem, we say that $L$ is a \dft{valid} labeling of $G$ if $(G, L) \in \calL$.

\begin{eg}
  \label{eg:vertex-labeling}
  We give three familiar examples of vertex labeling problems.
  
  \begin{enumerate}
  \item \dft{Proper $c$-coloring}. Take $\Sigma = [c] (= \set{1, 2, \ldots, c})$. Then $L \colon V \to \Sigma$ is a proper $c$-coloring if for all $uv \in E$, $L(u) \neq L(v)$.
  \item \dft{Independent set}. Take $\Sigma = \set{0, 1}$. Then $L \colon V \to \Sigma$ is an independent set if for all $uv \in E$ we have $L(u) L(v) = 0$---i.e., $L(u)$ and $L(v)$ are not both $1$. Note that taking $W = \set{v \in V \sucht L(v) = 1}$, we have $v \in W$ only if none of $v$'s neighbors are in $W$. Thus, our definition of independent set is equivalent to the more conventional definition.
  \item \dft{Downset.} Here we take $G$ to be a directed acyclic graph (DAG) and $\Sigma = \set{0, 1}$. Then $L$ is a downset if for every directed edge $uv \in E$, $L(v) = 1 \implies L(u) = 1$. As with our definition of independent set, we obtain the standard definition of downset by associating a valid vertex labeling $L$ with the set $W = \set{v \in V \sucht L(v) = 1}$. With this interpretation, $W$ is a downset in $G$ if and only if for every $v \in W$, every vertex $u$ from which $v$ is reachable (i.e., there is a directed path from $u$ to $v$), we also have $u \in W$.
  \end{enumerate}
\end{eg}

The three labeling problems in Example~\ref{eg:vertex-labeling} have a common feature: to check whether or not $(G, L) \in \calL$, it suffices to verify some Boolean predicate on each $uv \in E$ individually. We call such labeling problems ``edge-universal labeling problems.''

\begin{dfn}
  \label{dfn:edge-universal}
  Let $\calL$ be a vertex labeling problem. We say that $\calL$ is \dft{edge-universal} if there exists a Boolean predicate $\bfP \colon \Sigma \times \Sigma \to \set{0, 1}$ such that
  \[
  (G, L) \in \calL \iff \forall uv \in E, \bfP(L(u), L(v)) = 1.
  \]
  That is, $L$ is valid if and only if $\bfP$ is satisfied for each edge individually.
\end{dfn}

\subsection{Counting Algorithms}
\label{sec:counting-algorithms}

Let $\calL$ be an edge-universal labeling problem with labels $\Sigma$, and let $c = \abs{\Sigma}$. We describe a simple algorithm that for any graph $G$ counts the number of valid labelings $(G, L) \in \calL$ in time $O(c^{\pw + 1} \pw n)$, assuming a nice path decomposition $\calX$ of $G$ of width $\pw$ is given. The algorithm is a straightforward dynamic programming algorithm.

Suppose $\calX = (X_1, X_2, \ldots, X_{2n})$ is a nice path decomposition of $G$, and for completeness assume $X_0 = \varnothing$. For $i = 1, 2, \ldots, 2n$, let $F_i$ be the set of edges whose endpoints are in $X_i$:
\[
F_i = \set{uv \in E \sucht u, v \in X_i}.
\]
Recall that that $v$ is \emph{inserted at step $i$} if $X_{i} = X_{i - 1} \cup \set{v}$, and that $v$ is \emph{removed at step $i$} if $X_{i} = X_{i-1} \setminus \set{v}$. Define the sets $V_1, V_2, \ldots, V_{2n} (= V)$ by
\begin{equation*}
  V_i = \bigcup_{j \leq i} X_j.
\end{equation*}
Finally, let $G_i = (V_i, E_i) = G\restriction_{V_i}$, the subgraph of $G$ induced by $V_i$.

The algorithm we present computes the number of valid labelings of $G_i$ for each $i \in [2n]$.  More specifically, we maintain a counter $C_i(L_i)$ for each $i \in [2n]$ and each $L_i \colon X_i \to \Sigma$; it stores the number of valid labelings of $G_i$ whose restriction to vertices in $X_i$ is equal to $L_i$. That is,
\[
C_i(L_i) = \abs{\set{ L \colon V_i \to \Sigma \sucht L \text{ is a valid labeling of } G_i \text{ and } L(v) = L_i(v) \text{ for all } v \in X_i}}.
\]
We emphasize that $X_i \subseteq V_i$ so $L$ is an extension of $L_i$ when $X_i \subset V_i$ and $L = L_i$ when $X_i = V_i$.  Thus, the number of valid labelings of $G_i$ is $\sum_{L_i} C_i(L_i)$ (which contains $k^{\abs{X_i}} \leq k^{1+\pw}$ terms).  

Here, we define the rules that compute $C_{i}$ from $C_{i-1}$. Let $v_i$ denote the (unique) vertex added or removed at step $i$. The update procedure for $C_i$  has two cases, corresponding to $i$ being an insertion event and a removal event.

\begin{algorithm}
  \caption{$\CountValidLabelings(G, \calX)$}
  \begin{description}
  \item[Input.] $G = (V, E)$, a graph, $\calX = (X_1, X_2, \ldots, X_{2n})$ a nice path decomposition of $G$
  \item[Initialize.] $C_0(L_\varnothing) = 1$, where $L_\varnothing$ denotes the empty labeling
  \end{description}
  For $i = 1, 2, \ldots, 2n$, do the the following:
  \begin{description}
  \item[Insertion.] If $v$ is inserted at step $i$, for each $L_i \colon X_i \to \Sigma$, set $C_i(L_i) = C_{i-1}(L_i\restriction_{X_{i-1}})$ if all edges incident to $v$ in $F_i$ satisfy $\bfP$, and $C_i(L_i) = 0$ otherwise. That is,
    \begin{equation}
      \label{eqn:insert}
      C_i(L_i) = C_{i-1}(L_i\restriction_{X_{i-1}}) \prod_{w : vw \in F_i} \bfP(L_i(v), L_i(w)).
    \end{equation}
  \item[Removal.] If $v$ is removed at step $i$: for each $L_i \colon X_i \to \Sigma$,
    \begin{equation}
      \label{eqn:remove}
      C_i(L_i) = \sum_{\sigma \in \Sigma} C_{i-1}(L_i \cup \set{v \mapsto \sigma}).
    \end{equation}
  \end{description}
\end{algorithm}


The following lemma asserts the correctness of the formulas in~(\ref{eqn:insert}) and~(\ref{eqn:remove}).

\begin{lem}
  \label{lem:count-update}
  For each $i \in [2n]$ and $L_i \colon X_i \to \Sigma$, the recursive formulas in~(\ref{eqn:insert}) and~(\ref{eqn:remove}) imply  
  \[
  C_i(L_i) = \abs{\set{ L \colon V_i \to \Sigma \sucht L \text{ is a valid labeling of } G_i \text{ and } L(v) = L_i(v) \text{ for all } v \in X_i}}.
  \]
\end{lem}
\begin{proof}
  We argue by induction on $i$. For the base case $i = 1$, the vertex $v_1$ is added to $X_1$, and $V_1 = \set{v_1}$. By the description of the insertion procedure, we have
  \[
  C_i(\set{v_1 \mapsto \sigma}) = 1 \quad\text{for all } \sigma \in \Sigma.
  \]
  These are the correct values for a singleton graph, because $\calL$ is an edge-universal labeling problem and all edge predicates are vacuously satisfied (i.e., all vertex labelings are valid in a graph without edges).

  For the inductive step, assume the conclusion of the lemma holds for $i - 1$. We consider insertions and removals separately. First, suppose vertex $v_i$ is inserted at step $i$; i.e, $X_i = X_{i-1} \cup \set{v_i}$. Given a valid labeling $L$ on $G_{i-1}$, an extension $L \cup \set{v_i \mapsto \sigma}$ is valid if and only if $P(L(u), \sigma)$ is satisfied for each edge $uv_i \in F_i$. By the inductive hypothesis, for a fixed $\sigma$, the number of valid labelings $L'$ of $V_i$ satisfying $L'\restriction_{X_i} = L_i$ is therefore  $C_{i-1}(L_{i-1} \restriction_{X_{i} \setminus \set{v_i}})$ if $P$ is satisfied on all edges $uw \in F_i$, and $0$ otherwise. That is,~(\ref{eqn:insert}) correctly counts the number of valid labelings of $G_i$ whose restriction to $X_i$ is $L_i$.

  Now suppose $v_i$ is removed at index $i$, so that $X_i = X_{i-1} \setminus \set{v_i}$. Then $G_{i-1} = G_{i}$. Thus $L \colon V_i \to \Sigma$ is valid in $G_i$ if and only if it is valid in $G_{i-1}$. Therefore, for any $L_i \colon X_i \to \Sigma$ we compute
  \begin{align}
    \abs{\set{L \sucht L \restriction_{X_i} = L_i}} &= \sum_{\sigma \in \Sigma}\abs{\set{L \sucht K\restriction_{X_{i}} = L_i \text{ and } L(v_i) = \sigma}}\\
    &= \sum_{\sigma \in \Sigma} C_{i-1}(L_{i} \cup \set{v_i \mapsto \sigma}).
  \end{align}
  Thus~(\ref{eqn:remove}) correctly counts the number of valid labelings.
\end{proof}

\begin{lthm}
  \label{thm:counting}
  Let $G = (V, E)$ be a graph on $n$ vertices and $\calL$ an edge-universal vertex labeling problem with $\abs{\Sigma} = c$. Suppose a (simple) path decomposition $\calX$ of width $\pw$ is given and a look-up table for the predicate $\bfP : \Sigma \times \Sigma \to \set{0, 1}$ is pre-computed. Then the number of valid labelings of $G$ can be computed in time $O(c^{\pw + 1} \pw n)$. In particular, counting valid labelings can be performed in linear time for any family of graphs with bounded pathwidth (even if $\calX$ is not given in advance).
\end{lthm}
\begin{proof}
  Consider an execution of $\CountValidLabelings(G, \calX)$. Let $v = v_{2n}$ be the last node removed in $\calX$, and recall that $G_{2 n} = G$. Therefore we have
  \begin{align*}
    \abs{\set{L \sucht L \text{ is a valid labeling}}} &= \sum_{\sigma \in \Sigma} \abs{\set{L \sucht L \text{ is valid and } L(v) = \sigma}}\\
      &= \sum_{\sigma \in \Sigma} C_{2n}(\set{v \mapsto \sigma}).
  \end{align*}
  Thus, $\CountValidLabelings$ readily computes the number of valid labelings.

  For the runtime of $\CountValidLabelings$,  observe that at each step $C_i(L_i)$ is computed for every labeling $L_i: X_i \rightarrow \Sigma$.   To do so,  create a table with $c^{|X_i|}$ rows and $|X_i|+1$ columns.   Each row represents an $L_i$ labeling;  the first $|X_i|$ entries specifies the label assigned by $L_i$ to each vertex of $X_i$ while the last entry contains the value of $C_i(L_i)$.  Thus, the size of the table is $c^{|X_i|} (|X_i|+1) = O(c^{\pw+1} \pw)$.  
  
  It is easy to fill the first $|X_i|$ columns of the table as we just have to enumerate all the $c^{|X_i|}$ labelings.  The last column is filled by using the previous iteration's table and  formulas (\ref{eqn:insert}) and~(\ref{eqn:remove}).  In particular, when vertex $v_i$ is added to $X_i$, check if $C_{i-1}(L_i \restriction_{X_i \setminus{\set{v_i}}}) = 0$.  If yes, set $C_i(L_i) = 0$.  Otherwise, additionally check if $\bfP(L_{i}(v), L_{i}(w)) = 1$ for each neighbor $w$ of $v_i$ in $X_{i}$.  If yes, set $C_i(L_i) = C_{i-1}(L_i \restriction_{X_i \setminus{\set{v_i}}})$; if not, set $C_i(L_i) = 0$.  This will take $O(|X_i| -1) = O(\pw)$ time since $v_i$ can have at most $|X_i| -1$ neighbors in $X_i$, and this operation is performed at most once per row. Thus, when $v_i$ is added to $X_i$, creating and filling the table takes $O(c^{\pw+1} \pw)$ time.
  
   On the other hand, when $v_i$ is removed at index $i$ so that $X_{i} = X_{i-1} \setminus{\set{v_i}}$, then $C_i(L_i)$ is obtained by looking up the values of $C_{i-1}(L_i \cup \set{v_i \mapsto \sigma})$.  Notice that filling \emph{all} the $C_i(L_i)$ requires looking up each $C_{i-1}(L_{i-1})$ entry exactly once.  Hence, when $v_i$ is removed at index $i$, creating and filling the table takes $O(c^{\pw+1} \pw)$ time.   
   
   Combining our analysis,  the total runtime is $\sum_{i} O(c^{\pw + 1} \pw) = O(c^{\pw + 1} \pw n)$, as claimed.   Finally, given $G$, a nice path decomposition $\calX$ of width $\pw(G)$ can be computed in FPT linear time using, for example, the algorithm of Bodelaender~\cite{Bodlaender1996-linear-time} (cf. Theorem~\ref{thm:pathwidth}).  
\end{proof}

For the problems listed in Example~\ref{eg:vertex-labeling}, constructing the look-up table for the predicate $\bfP : \Sigma \times \Sigma \to \set{0, 1}$ takes $O(c^2)$ time.  Applying Theorem~\ref{thm:counting}, we obtain the following corollary. 

\begin{cor}
  \label{cor:counting}
  Given a graph $G$ on $n$ vertices and a nice path decomposition $\calX$ of $G$ of width $\pw$, then we can compute:
  \begin{enumerate}
  \item the number of $c$ colorings of $G$ in time $O(c^{\pw + 1} \pw n)$;
  \item the number of independent sets in $G$ in time $O(2^{\pw} \pw n)$;
  \item the number of downsets in $G$ in time $O(2^{\pw} \pw n)$ if $G$ is a DAG. \label{it:downsets}
  \end{enumerate}
\end{cor}

\begin{rem}
  Our algorithm for counting $c$ colorings can easily be employed to compute the chromatic number of a graph. Since the chromatic number of a graph is at most one more than its pathwidth, the runtime of such an algorithm is $O((\pw + 1)^{\pw + 2} n)$. For graphs with small pathwidth, this is faster than worst-case exponential-time algorithms for chromatic number, such as~\cite{Lawler1976-note}.
\end{rem}

\subsubsection{Generalization}

Here, we describe a straightforward generalization of $\CountValidLabelings$ that counts \emph{extensions} of a \emph{partial} labeling $K$. We will require this generalization as a subroutine in the sequel.

\begin{lem}
  \label{lem:count-extensions}
  Let $\calL$ be an edge universal labeling problem. Then there exists an algorithm $\CountExtensions$ such that for any graph $G$, nice path decomposition $\calX$ of width $\pw$, and partial labeling $K$, $\CountExtensions(G, \calX, K)$ computes the number of valid extensions $L$ of $K$ in time $O(k^{\pw + 1} \pw n)$.
\end{lem}
\begin{proof}[Proof sketch.]
  Fix a graph $G = (V, E)$, nice path decomposition $\calX = (X_1, X_2, \ldots, X_{2n})$ and partial labeling $K$ of $G$. Let $A \subseteq V$ denote the set of assigned vertices for $K$. For $i = 1, 2, \ldots, 2n$, let $K_i$ denote the restriction of $K$ to $X_i$. Consider the modification of $\CountValidLabelings$ in which $C_i(L_i)$ is computed only for labelings $L_i \colon X_i \to \Sigma$ extending $K_i$. An argument analogous to our proof of Lemma~\ref{lem:count-update} shows that for all $i$ and extensions $L_i$ of $K_i$, we have
  \[
  C_i(L_i) = \abs{\set{L \colon V_i \to \Sigma \sucht L \text{ a valid labeling of } G_i,  L(v) = K(v) \text{ for all } v \in V_i \cap A \text{, and } L\restriction_{X_i} = L_i}}.
  \]
  That is, $C_i(L_i)$ counts the number of valid labelings of $G_i$ extending $K$ (restricted to $G_i$) whose restriction to $X_i$ is $L_i$. In particular, taking $i = 2n$, we have $G_{2n} = G$ so that $C_{2n}(\varnothing)$ gives the number of valid extensions of $K$. The runtime analysis of $\CountExtensions$ is identical to the proof of Theorem~\ref{thm:counting}.
\end{proof}

\subsection{Sampling Algorithms}
\label{sec:sampling-algorithms}

In this section, we show how $\CountExtensions$ can be used as a sub-routine in order to sample valid labelings uniformly from any edge universal labeling problem $\calL$. The idea is as follows. Fix an (arbitrary) ordering of the vertices $v_1, v_2, \ldots, v_n$. We then form a labeling $L$ by sequentially fixing $L(v_1), L(v_2), \ldots, L(v_n)$ in such a way that $L$ is chosen uniformly at random. 

In more detail, our sampling algorithm forms partial labelings $K_1, K_2, \ldots, K_n = L$, where in each $K_i$, $v_1, v_2, \ldots, v_i$ are assigned, while the other vertices are unassigned. For completeness, fix $K_0$ to be the partial labeling in which all vertices are unassigned. $K_i$ is determined from $K_{i-1}$ by setting $K_i(v_j) = K_{i-1}(v_j)$ for all $j < i$, and $K_i(v_i) = \sigma_i$, where $\sigma_i \in \Sigma$ is chosen in proportion to the number of valid extensions $L$ of $K_{i-1}$ satisfying $L(v_i) = \sigma_i$. The sampling procedure $\SampleLabeling$ (Algorithm~\ref{alg:sample-labeling}) formalizes the sampling procedure.

\begin{algorithm}
  \caption{$\SampleLabeling(G = (V, E), \calX, \calL)$ samples a uniformly random valid labeling of $G$ with respect to $\calL$. We assume that $G$ admits at least one valid labeling (which can be checked using $\CountValidLabelings$).}
  \label{alg:sample-labeling}
  \begin{algorithmic}[1]
    \STATE initialize $K(v) \leftarrow \perp$ for all $v \in V$
    \FOR{$v_1, v_2, \ldots, v_n \in V$}
    \FORALL{$\sigma \in \Sigma$}
    \STATE $c_{\sigma} \leftarrow \CountExtensions(G, \calX, K \cup \set{v \mapsto \sigma})$
    \ENDFOR
    \STATE $c \leftarrow \sum_{\sigma \in \Sigma} c_{\sigma}$ \COMMENT{note that $c = \CountExtensions(G, \calX, K)$}
    \STATE choose random $\tau \in \Sigma$ with $\Pr(\tau = \sigma) = c_\sigma / c$ \label{ln:tau}
    \STATE $K \leftarrow K \cup \set{v \mapsto \tau}$
    \ENDFOR
    \RETURN $K$
  \end{algorithmic}
\end{algorithm}

We now prove the main result of this section.

\begin{lthm}
  \label{thm:sampling}
  For any graph $G = (V, E)$ on $n$ vertices, nice path decomposition $\calX$ of $G$ of width $\pw$, and edge universal labeling problem $\calL$, $\SampleLabeling(G, \calX, \calL)$ returns a uniformly random valid labeling of $G$ with respect to $\calL$. The runtime of $\SampleLabeling$ is $O(k^{\pw + 2} \pw n^2)$. 
\end{lthm}
\begin{proof}
  Let $L$ be a fixed valid labeling of $G$, and let $K$ denote the random variable labelings returned by $\SampleLabeling$. For $j = 1, 2, \ldots, n$, we denote $\sigma_j = L(v_j)$, and take $L_j$ to be the restriction of $L$ to $\set{v_1, \ldots, v_j}$ (with $v_i$ unassigned for $i > j$). We set $L_0$ to be the empty labeling with all vertices unassigned. We compute
  \begin{align*}
    \Pr(K = L) &= \Pr(K(v_1) = \sigma_1, K(v_2) = \sigma_2, \ldots, K(v_n) = \sigma_n)\\
    &= \prod_{j = 1}^n \Pr\paren{K(v_j) = \sigma_j \sucht K \text{ extends } L_{j-1}}\\
    &= \prod_{j = 1}^n \frac{\abs{\set{J  \sucht J \text{ is a valid extension of } L_{j}}}}{\abs{\set{J \sucht J \text{ is a valid extension of } L_{j-1}}}}\\
    &= \frac{1}{\abs{\set{J \sucht J \text{ is a valid labeling}}}}.
  \end{align*}
  The third equality holds by the choice of $\tau$ in Line~\ref{ln:tau} and the result of Lemma~\ref{lem:count-extensions}. The final equality is due to the telescoping product. Thus, every valid labeling $L$ is returned with equal probability.

  Finally the assertion about the runtime of $\SampleLabeling$ follows from runtime of $\CountExtensions$ (Lemma~\ref{lem:count-extensions}), and the observation that $\CountExtensions$ gets called once during each of the $k \cdot n$ iterations of the inner loop in $\SampleLabeling$.
\end{proof}

\begin{cor}
  \label{cor:sampling}
  Given a graph $G$ on $n$ vertices and a nice path decomposition $\calX$ of $G$ of width $\pw$, then we can sample a uniformly random:
  \begin{enumerate}
  \item $c$ coloring of $G$ in time $O(c^{\pw + 2} \pw m n)$;
  \item independent set in $G$ in time $O(2^{\pw} \pw m n)$;
  \item downset in $G$ in time $O(2^{\pw} \pw m n)$ if $G$ is a DAG. \label{it:sample-downsets}
  \end{enumerate}
\end{cor}

\section{Counting and Sampling Elements in a Distributive Lattice}

In what follows, we describe some applications of the counting and sampling algorithms from the previous sections.

   A \dft{distributive lattice} $\calD = (D, \leq)$ is a partially ordered set where any pair of elements $x$ and $y$ has a (i) \dft{meet or greatest lower bound} $x \land y$, (ii) a \dft{join or least upper bound} $x \lor y$ and (iii) the meet and join operations distribute over each other.  Many combinatorial objects are known to form a distributive lattice including the minimum cuts in a network~\cite{Picard1980-structure}, circulations in a planar graph~\cite{Khuller1993-lattice} domino tilings of a polygon, the perfect matchings of a bipartite planar graph, alternating sign matrices, flows in planar graphs etc. are known to form a distributive lattice~\cite{Propp1997-generating}.   

Here is an easy way of creating a distributive lattice from an arbitrary poset $P$.  Find all the downsets of $P$ and order them using the subset relation.  It is straightforward to verify that given two any downsets $S$ and $S'$ of $P$,   $S \cap S'$ and $S \cup S'$ are also downsets of $P$ and are in fact the greatest lower bound and the least upper bound respectively of $S$ and $S'$.   Since $\cap$ and $\cup$ distribute over each other, it follows that $(D(P), \subseteq)$ is a distributive lattice, where $D(P)$ contains all the downsets of $P$.  Interestingly,   Birkhoff~\cite{Birkhoff1937-rings} showed that downsets are integral to distributive lattices than one might initially suspect.

\begin{lthm}[Birkhoff~\cite{Birkhoff1937-rings}] For every distributive lattice $\calD$, there is (up to isomorphism) a unique poset $P_{\calD}$ so that $D(P_{\calD}, \subseteq)$ is a distributive lattice that is order isomorphic to $\calD$. 
\end{lthm}

In Birkhoff's proof, the poset $P_{\calD}$ was formed using the \dft{join-irreducible elements} of $\calD$; i.e.,  the elements that have an in-degree of $1$ in the Hasse diagram of $\calD$.  Thus, it seems that $\calD$ has to be constructed first to obtain $P_{\calD}$.  In practice, however, $P_{\calD}$ can sometimes be computed \emph{directly} from the problem description. This is the case for stable matchings. A typical instance $I$ consists of $n$ men and $n$ women, each with their own preference lists.  It has long been known that the set of stable matchings of $I$ form a distributive lattice. But it was not until the mid-1980's that Irving and Leather~\cite{Irving1986-complexity} showed that the corresponding poset of the lattice can be computed directly from $I$'s man-optimal stable matching and preference lists in polynomial time.  They called it the \dft{rotation poset} of $I$.  Remarkably, $I$ can have an exponential number of stable matchings~\cite{Karlin2018-simply} (so the distributive lattice can have exponential size as well) but the rotation poset is guaranteed to have just $O(n^2)$ elements.  Gusfield~\cite{Gusfield1987-three} further improved Irving and Leather's result by showing that a direct acyclic graph $G(I)$ can be constructed in $O(n^2)$ time so that $G(I)$ and the rotation poset of $I$ have exactly the same transitive closure and, therefore, the same downsets.  

Thus, the number of stable matchings of a stable matching instance $I$ can be obtained by first computing $G(I)$ and then returning the number of downsets of $G(I)$.  Similarly, uniformly sampling a stable matching of $I$ can be done by first computing $G(I)$, uniformly sampling a downset and then returning the stable matching that corresponds to this downset.   We have established the result below.  

\begin{lthm}
  Let $I$ be a stable matching instance with $n$ men and $n$ women. Then
  \begin{itemize}
  \item[(i)] the number of stable matching of $I$ can be computed in $O(f(\pw) n^2)$ time where $\pw$ denotes the pathwidth of $G(I)$, and
  \item[(ii)] a stable matching of $I$ can be sampled uniformly at random in $O(f(\pw) n^4)$ time.
  \end{itemize}
\end{lthm}

We note that the approach of using the underlying poset $P_\calD$ to study the distributive lattice $\calD$ is not unique to stable matchings. It has been considered for circulations in planar graphs~\cite{Khuller1993-lattice}, area-universal rectangular layouts~\cite{Eppstein2009-area}, etc. Applying the same reasoning above to arbitrary finite distributive lattices, we also have the following result.

\begin{lthm}
  Let $\calD$ be a finite distributive lattice.  Given $P_\calD$ (directed acyclic graph $G_\calD$ whose transitive closure is identical to $P_\calD$) that has $n$ vertices and pathwidth $k$, then
  \begin{itemize}
  \item[(i)] the number of elements of $\calD$ can be computed in $O(f(\pw) n)$ time,
    while
  \item[(ii)] uniformly sampling an element of $\calD$ can be done in $O(f(\pw) n^2)$ time.
  \end{itemize}
  Here $n$ and $\pw$ denote the number of vertices and pathwidth of $G_\calD$, respectively.
\end{lthm}

Sampling algorithms are investigated for a variety of reasons.  Propp~\cite{Propp1997-generating}, for example, detected a ``circular" phenomenon in $(a,b,c)$-partitions by generating random tilings. For stable matchings, uniform sampling is a way to obtain a ``fair'' stable matching~\cite{Bhatnagar2008-sampling} as the man-optimal and woman-optimal stable matchings are also woman-pessimal and man-pessimal stable matchings, respectively.  

In a recent paper~\cite{Cheng2020-stable},  we considered stable matching instances with ``$k$-range preferences.''  That is, there is an objective ranking for each group of agents, and each person ranks agents from the other group to within $k$ of their objective ranks.  This model captures the scenario when participants make use of ``official rankings'' to create their preference lists. The smaller the value of $k$, the more faithful the participants' rankings are to the official rankings.  

In general, \emph{every} stable matching instance has $k$-range preferences for some $k \leq n$. Finding the smallest such $k$ can be done in $O(n^2)$ time. We refer to it as the \dft{range of $I$} and denote it as $\range(I)$. We proved the following theorem.  

\begin{lthm}[\cite{Cheng2020-stable}]
  \label{thm:k-range}
  Let $I$ be a stable matching instance with $n$ men and $n$ women.  Suppose $\range(I) = k$. Then $G(I)$ and a path decomposition $\calX$ of $G(I)$  of width $\pw = O(k^2)$ can be computed $O(k^2 n + n^2)$ time.  
\end{lthm}

\begin{cor}
  \label{cor:sm}
  Let $I$ be a stable matching instance with $\range(I) = k$. Then the number of stable matchings in $I$ can be computed in time $O(2^{O(k^2)} k^2 n + n^2)$, and a uniformly random stable matching can be found in time $O(2^{O(k^2)} k^4 n^2)$. In particular, counting and uniformly sampling stable matchings are fixed parameter tractable when parameterized by the range of the instance.
\end{cor}

We note that for fixed $k$, the runtimes above are \emph{linear} in the size of the stable matching instance, as specifying $k$-range preferences for $\Theta(n)$ agents requires $\Theta(n^2)$ bits, even for constant $k$.

\urlstyle{same}
\bibliographystyle{plainnat}
\bibliography{counting-sampling-pathwidth}

\end{document}